\documentclass[12pt,reqno]{amsart}

\usepackage{amsaddr,amssymb}
\usepackage{graphicx} 
\usepackage{caption}
\usepackage{subcaption}

\newtheorem{claim}{Claim}[section]
\newtheorem{theorem}[claim]{Theorem}
\newtheorem{proposition}[claim]{Proposition}

\newtheorem{definition}[claim]{Definition}

\numberwithin{equation}{section}

\begin{document}
\title[Quotient graphs]{Application of quotient graph theory to three-edge star graphs}
\author{Vladim\'ir Je\v{z}ek} 
\address{Department of Physics, Faculty of Science, University of Hradec Kr\'alov\'e, Rokitansk\'eho 62,
500\,03 Hradec Kr\'alov\'e, Czechia}
\email{vljezek@atlas.cz}

\author{Ji\v{r}\'{\i} Lipovsk\'{y}}
\address{Department of Physics, Faculty of Science, University of Hradec Kr\'alov\'e, Rokitansk\'eho 62,
500\,03 Hradec Kr\'alov\'e, Czechia}
\email{jiri.lipovsky@uhk.cz}

\begin{abstract}
We apply the quotient graph theory described by Band, Berkolaiko, Joyner and Liu to particular graphs symmetric with respect to $S_3$ and $C_3$ symmetry groups.  We find the quotient graphs for the three-edge star quantum graph with Neumann boundary conditions at the loose ends and three types of coupling conditions at the central vertex (standard, $\delta$ and preferred-orientation coupling). These quotient graphs are smaller than the original graph and the direct sum of quotient graph Hamiltonians is unitarily equivalent to the original Hamiltonian.
\end{abstract}

\maketitle



\section{Introduction}
Symmetry plays an important role in many branches of physics and mathematics. It can be found in many systems in physics, chemistry or biology, as crystals, molecules, living organisms, or the structure of fundamental laws of nature. Its importance lies in simplifying the tasks; even a very difficult problem can be reduced, using its symmetry, and solved significantly more easily. 

The quantum graphs, first used for the description of aromatic molecules in the 1930's \cite{Pauling} and 1950's \cite{RuedenbergScherr}, then widely studied since the 1980's, can serve as a nice example of the importance of symmetry. This model, reasonably simple from the mathematical point of view (set of ordinary differential equations), shows many non-trivial properties and therefore is used as a toy model, e.g. for describing quantum chaos \cite{KottosSmilansky,BerkolaikoBogomolny}. Quantum graphs, however, are not an artificial problem; the Schr\"odinger equation on a network has applications in describing nanotubes, photonic crystals, etc. The mathematical claims on the properties of this quantum problem, do not need quantum theory to be experimentally verified. Using similar forms of the Schr\"odinger and telegraph equation, one can model quantum graphs with the so-called microwave graphs -- the behavior of a quantum particle is replaced by the propagation of microwaves in coaxial cables \cite{HulEtAl, HulEtAl2, LawniczakLipovskySirko, LawniczakKurasov, LawniczakLipovskyBialousSirko}.

Symmetry allows decomposing complicated graphs with many edges into simpler graphs, for which the term \emph{quotient graphs} is used. This is useful e.g. for finding the secular equation for the graph eigenvalues or the resonance condition for resolvent resonances. As we show in Section~\ref{sec:qg}, the secular equation is given by the determinant of a square matrix with the number of rows and columns being double the number of graph edges. Hence reducing the number of edges significantly simplifies the computation. 

The paper~\cite{BandBerkolaikoJoynerLiu} summarized the theory (previously developed in \cite{BandParzanchevskiBen-Shach, BandParzanchevski}) for constructing quotient graphs using the symmetry groups of the graph. Applications to both combinatorial and quantum graphs are provided in the mentioned paper. Using this construction, one may obtain the quotient graphs, each corresponding to one particular irreducible representation of the symmetry group. There are various utilizations of this theory, it was used for simplifying the graph and computing the secular equation e.g. in \cite{ExnerLipovsky1}. The theory was also applied to the construction of quantum graphs providing GSE (Gaussian Symplectic Ensemble) statistics in \cite{JoynerMullerSieber}.

The present paper aims to introduce the quotient graph theory developed in~\cite{BandBerkolaikoJoynerLiu} in a compact form and showing its applications in rather simple, but still non-trivial examples. We focus on quantum graphs quotients only, in particular, we chose equilateral star graphs consisting of three edges. Alternating the coupling condition at the central vertex, we can change the symmetry group of the graph. In detail, we show the construction for the $S_3$ group (the symmetry group for standard and $\delta$-coupling), including the construction of the irreducible representations of this group. Introducing a preferred direction in the graph using a special type of coupling first used in \cite{ExnerTater}, one may reduce the group symmetry to $C_3$. Hence the irreducible representations change and instead of two one-dimensional and one two-dimensional representations we obtain three one-dimensional representations.

The paper is structured as follows. In the next two sections, we give necessary preliminaries needed for stating the theorem of~\cite{BandBerkolaikoJoynerLiu} -- in Section~\ref{sec:group} we introduce the main notions of the group theory, and Section~\ref{sec:qg} is devoted to quantum graphs. In Section~\ref{sec:quotient} we state the procedure from~\cite{BandBerkolaikoJoynerLiu} allowing us to obtain quotients for the quantum graphs. In Section~\ref{sec:three-edge} we apply this theory to three-edge graphs. We obtain the representations of the $S_3$ group and find the kernel space needed in the procedure. Using it, we find in Subsections~\ref{sec:standard} and \ref{sec:delta} the quotient graphs for standard and $\delta$-coupling. Subsection~\ref{sec:preferred} is devoted to the example of the graph with preferred-orientation coupling which is symmetric under the $C_3$ group.  Finally, we conclude the results in Section~\ref{sec:conclusions}. 

\section{Preliminaries about group theory}\label{sec:group}
In this section, we revise necessary notions of the group theory, which allow us to formulate the quotient graph method. We focus mainly on the representation theory for groups. The current paper cannot give a full and detailed description of the field, therefore, we refer the interested reader e.g. to publications \cite{FultonHarris, BarutRaczka, Cotton}. We start with the definition of the group.

\begin{definition}
Group $(G,\cdot)$ is a set $G$ with a binary operation ``$\cdot$'' for which the following properties hold
\begin{enumerate}
\item $G$ is closed with respect to ``$\cdot$'', i.e. $a \cdot b \in G$ for all $a, b \in G$,
\item the operation ``$\cdot$'' is associative, i.e.  $(a \cdot b)\cdot c = a\cdot (b\cdot c)$ for all $a, b, c \in G$,
\item there exists an identity element $e$ for which $e \cdot a = a \cdot e$ for all $a \in G$,
\item for each $a\in G$ there exists the inverse element $a^{-1}$ such that $a^{-1}\cdot a = a\cdot a^{-1} = e$.
\end{enumerate}
\end{definition}

\begin{definition}
Two elements $a, b \in G$ are in the same \emph{conjugacy class} if there is an element $g\in G$, such that $b = g^{-1} \cdot a \cdot g$.
\end{definition}

\begin{definition}
Let $(G,*)$ and $(H,\cdot)$ be two groups. The map $\varphi: G\to H$ is called a \emph{homomorphism} from the group $G$ to $H$ if it satisfies
$$
   \varphi(a * b) = \varphi(a)\cdot \varphi(b)\,.
$$
If a homomorphism is bijective, we call it \emph{isomorphism}. Isomorphism $\varphi: G\to G$ is called \emph{automorphism}.
\end{definition}

Note that ``$*$'' is the group operation in the group $G$ and ``$\cdot$'' is the group operation in the group $H$. The homomorphism, therefore, preserves the group operation. In the following text, we will consider finite groups (groups with a finite number of elements).

\begin{definition}
Let $V$ be a vector space. The representation of the group $(G,*)$ on $V$ is a map $\rho : G\to GL(V)$ such that 
$$
  \rho(a*b) = \rho(a)\cdot \rho(b)
$$
for all $a, b \in G$. Here $GL(V)$ is the general linear group of the vector space $V$, i.e. the group of all automorphisms on $V$. The dimension of the vector space $V$ is called the \emph{dimension of the representation} or its \emph{degree}. The space $V$ is called the \emph{carrier space} of the representation.
\end{definition}

Let us briefly comment on the previous definition. The elements of the linear group $GL(V)$ can be viewed as square matrices; the corresponding group operation is matrix multiplication. This allows us to obtain an equivalent group to $(G,*)$, where the elements of the new group are square $d\times d$ matrices and the group operation is the matrix multiplication; here, $d$ is the dimension of the representation.

\begin{definition}
Let $(G,*)$ be a group and $\rho$ be its representation. A linear subspace $W\subset V$ is called $G$-invariant if $\rho(a) \cdot w \in W$ for all $a\in G$ and all $w\in W$. Here, ``$\cdot$'' is matrix multiplication between the matrix $\rho(a)$ and the finite-dimensional column vector $w$. If $V$ contains a subspace $W\subsetneq V$ with the previously mentioned property, we call the representation $\rho$ \emph{reducible}. Otherwise, it is called \emph{irreducible}.
\end{definition}

The meaning of the previous definition is that if we find a subspace for which all the matrices $\rho(a)$, $a\in G$ map this vector subspace to itself, we have a reducible representation. In other words, there exists a similarity transformation of all the matrices $\rho(a)$ which maps them into matrices of the block type $\begin{pmatrix}D^W & D^{WW'}\\ 0 & D^{W'}\end{pmatrix}$. Moreover, if the block $D^{WW'}=0$, the representation is called decomposable, as the next definition states.

\begin{definition}
The representation $\rho$ is called decomposable if there is a basis in which the matrices $\rho(a)$ are of block diagonal form. Then the subspace $W\subset V$ is called the reducing subspace.
\end{definition}

Each decomposable representation can therefore be written as a direct sum of two (or more) irreducible representations, each of them given by the matrices in blocks.

\begin{definition}
Let $V$ be a finite dimensional vector space over a field $T$, $(G,*)$ a group and $\rho$ a representation of $(G,*)$ on $V$. Then the function $\chi_\rho : G \to T$ defined as
$$
  \chi_\rho(a) = \mathrm{Tr\,}(\rho(a))
$$
for each $a\in G$ is called a \emph{character} of the representation $\rho$. Here, the symbol $\mathrm{Tr}$ denotes the trace of a matrix. By $\chi_\rho(a)$ we mean the character of the element $a$ in the representation $\rho$.
\end{definition}
 
\begin{proposition}\label{thm:conjclass}
The characters of the group elements in the same conjugacy class are the same. 
\end{proposition}
\begin{proof}
Clearly, using the properties of the trace, we have 
$$
 \mathrm{Tr}\,(g^{-1}\cdot a \cdot g) = \mathrm{Tr}\,(a) 
$$
which proves the claim.
\end{proof}

\begin{definition}
Let $|G|$ be the number of elements of the group $G$. For characters $\chi_{\rho_1}$, $\chi_{\rho_2}$ of two representations $\rho_1$, $\rho_2$ we define the inner product as
$$
  \left<\chi_{\rho_1},\chi_{\rho_2}\right> := \frac{1}{|G|}\sum_{a\in G} \overline{\chi_{\rho_1}}(a) \chi_{\rho_2}(a)\,.
$$
\end{definition}

We state the following proposition, the proof can be found, e.g. in \cite{FultonHarris}.
\begin{proposition}\label{thm:inner}
The following properties of the inner product hold.
\begin{enumerate}
\item[i)] The representation $\rho$ is irreducible if and only if its character $\chi_\rho$ satisfies $\left<\chi_\rho,\chi_\rho\right> = 1$. 
\item[ii)] Let $V_j$, $j = 1,\dots k$ be the vector spaces associated with the irreducible representations $\rho_j$ and $V$ be the vector space associated with the reducible representation $\pi$ of the group $(G,\cdot)$. Let $V\cong V_1^{\oplus \alpha_1}\oplus V_2^{\oplus \alpha_2}\oplus \dots \oplus V_k^{\oplus \alpha_k}$ (the sign $\cong$ denotes isomorphism, $\oplus$ denotes the direct sum, and $V_1^{\oplus \alpha_1}$ means $\alpha_1$ copies of the vector space $V_1$). Then the multiplicity $\alpha_j$ of the irreducible representation $\rho_j$ in $\pi$ is given by
$$
  \alpha_j = \left<\chi_\pi, \chi_{\rho_j}\right>\,.
$$
\end{enumerate}
\end{proposition}

\begin{definition}\label{def:action}
Let $(G,*)$ be a group with the identity element $e$ and let $S$ be a set. Then the (left) action of $G$ on $S$ is the operation $\circ : G\times S \to S$ satisfying the following three axioms
\begin{enumerate}
\item[a)] $g \circ s \in S$ for all $s\in S$ and $g\in G$, 
\item[b)] $e\circ s = s$ for all $s\in S$,
\item[c)] $g_1\circ (g_2\circ s) = (g_1 * g_2) \circ s$ for all $s\in S$ and $g_1, g_2\in G$.
\end{enumerate}
\end{definition}

In the following text, we will, with small abuse of notation, denote the group action by the same symbol ``$*$'' as the group multiplication.

\section{Preliminaries about quantum graphs}\label{sec:qg}
We briefly introduce the usual description of quantum graphs. For more details, we refer the reader to the publications \cite{BerkolaikoKuchment, GnutzmannSmilansky}. 

Let us consider a metric graph consisting of $|\mathcal{V}|$ vertices and $|\mathcal{E}|$ edges $e_j$ of finite lengths $\ell_j$, $j = 1,\dots, |\mathcal{E}|$ that connect two vertices. The vertex set is denoted by $\mathcal{V}$ and the edge set by $\mathcal{E}$. We consider the Hilbert space $\mathcal{H} = \oplus_{j = 1}^{\mathcal{|E|}} L^2(e_j)$. In this Hilbert space we define a second-order differential operator $H$ acting as $H = -\frac{\mathrm{d}^2}{\mathrm{d}x^2}$ with the domain consisting of the functions with edge components in the Sobolev space $W^{2,2}(e_j)$ satisfying the coupling conditions at each vertex $\mathcal{X}_s\in \mathcal{V}$ with the degree (valency) $d_s$
\begin{equation}
  A_s \Psi_s + B_s\Psi'_s = 0\,,\label{eq:qg:coup1}
\end{equation}
where $A_s$ and $B_s$ are $d_s\times d_s$ matrices satisfying $A_s \cdot B_s^\dagger = B_s \cdot A_s^\dagger$ ($\dagger$ denotes the hermitian conjugation) and the joined rectangular matrix $(A_s,B_s)$ has maximal rank. The vector $\Psi_s$ is the vector of the limiting values of functions at the vertex $\mathcal{X}_s$ from the edges incident to this vertex and $\Psi_s'$ is a similarly defined vector of outgoing derivatives. 

The coupling on the whole graph can be described by the $2|\mathcal{E}|\times 2|\mathcal{E}|$ matrices $A$ and $B$ that can be obtained from block matrices consisting of $A_s$ and $B_s$, respectively, after a transformation that interchanges rows and columns. The coupling conditions \eqref{eq:qg:coup1} can be written in one equation 
\begin{equation}
  A \Psi + B\Psi'= 0\,. \label{eq:qg:coup2}
\end{equation}
Here, the vectors are
\begin{eqnarray*}
\Psi = (f_1(0),f_1(\ell_1), f_2(0), f_2(\ell_2),\dots, f_{\mathcal{|E|}(\ell_{\mathcal{|E|}})})^\mathrm{T}\,,\\
\Psi' = (f_1'(0),-f_1'(\ell_1), f_2'(0), -f_2'(\ell_2),\dots, -{f_{\mathcal{|E|}}'(\ell_{\mathcal{|E|}})})^\mathrm{T}\,,\\
\end{eqnarray*}
where $f_j$ are the components of the wavefunction on the edges of the graph.

The operator defined in the above manner is the Hamiltonian of a quantum particle on the graph in the set of units with $\frac{\hbar^2}{2m} = 1$ which moves freely on the graph edges and interacts only at the vertices. The properties of the matrices $A_s$ and $B_s$ ensure that the Hamiltonian is self-adjoint. Similarly, the matrices $A$ and $B$ satisfy $A \cdot B^\dagger = B \cdot A^\dagger$ and the maximal-rank condition. There is an alternative description of the coupling conditions using a unitary matrix $U$ (unitarity means the condition $U\cdot U^\dagger = U^\dagger U  = I$, where $I$ is an identity matrix). Since the whole equation \eqref{eq:qg:coup2} can be multiplied by a~regular square matrix from the left without changing the coupling condition, one can choose $A = C(U-I)$, $B = i C (U+I)$ with $C$ being a regular $2\mathcal{|E|}\times 2\mathcal{|E|}$ square matrix. Unitarity of $U$ results in satisfying the conditions on $A$ and $B$ and the Hamiltonian is therefore self-adjoint.

From the mathematical point of view, a quantum graph is a set of ordinary differential equations (ODE) coupled by vertex conditions. When finding the spectrum of the graph, one has to solve the eigenvalue equation $-f_j''(x) = k^2 f_j(x)$ at each edge of the graph. It follows from the ODE theory that the solutions $f_j$ can be found in the form $f_j(x) = a_j\sin{(kx)}+ b_j\cos{(kx)}$. Thus the energies $k^2$ can be found when one substitutes the above form of the wavefunctions into the coupling condition \eqref{eq:qg:coup2}, and constructs the secular equation given by vanishing the determinant of the matrix multiplying the vector of coefficients $(a_1,b_1,\dots, a_\mathcal{|E|}, b_\mathcal{|E|})^\mathrm{T}$.

\section{Quotient graph theory}\label{sec:quotient}
The procedure for obtaining the quotient graphs from the original quantum graph was described in \cite{BandBerkolaikoJoynerLiu}. We briefly describe its main concepts; for the proof and more insight, we refer to the mentioned publication. 

First, we introduce the Kronecker product.

\begin{definition}
The Kronecker product of two matrices $C$ ($m\times n$ matrix) and $D$ ($p\times q$ matrix) is a $mp \times nq$ matrix given by
$$
  C\otimes D  := \begin{pmatrix}c_{11} D & c_{12} D & \dots & c_{1n} D\\
								c_{21} D & c_{22} D & \dots & c_{2n} D\\
								\vdots & \vdots & \ddots & \vdots\\
								c_{m1} D & c_{m2} D & \dots & c_{mn} D
\end{pmatrix} \,,
$$
where $c_{ij}$, $i =1,\dots, m$, $j = 1,\dots, n$ are the entries of the matrix $C$. In the above equation there are denoted the $p\times q$ blocks of the resulting matrix.
\end{definition}

Secondly, we introduce the notion of a $\pi$-symmetric graph. Let us consider a quantum graph $\Gamma$ with finitely many finite edges $e_j$. Let $(G,*)$ be the symmetry group of the graph $\Gamma$ which maps each edge $e_j$ to another edge $g*e_j$, where ``$*$'' now denotes the group action on a set. The edge $g* e_j$ may or may not be the same one, however, we assume that $G$ does not map any edge to its reverse. In that case, we would introduce a vertex with the standard condition in the middle of this edge and thus dividing it into two. 

\begin{definition}
Let $\pi: G\to GL(\mathbb{C}^{|\mathcal{E}|})$ be a representation of a group $(G,*)$ such that for each $g\in G$ the matrix $\pi(g)$ is a permutation matrix. The graph $\Gamma$ is $\pi$-symmetric if the following two conditions hold
\begin{enumerate}
\item For each $g\in G$ and each $j=1,\dots, |\mathcal{E}|$ and the index $i$ given by $e_i = g* e_j$, the condition $\ell_{j} = \ell_{i}$ holds. 
\item The coupling condition \eqref{eq:qg:coup2} for the coupling matrices $A$ and $B$ is satisfied iff this coupling condition is satisfied for each $g\in G$ for the coupling matrices $A\cdot \hat \pi(g)$ and $B\cdot \hat \pi(g)$ (dot denotes matrix multiplication), where $\hat \pi(g) = \pi(g)\otimes I_2$ (here $I_2$ denotes the $2\times 2$ identity matrix).
\end{enumerate}
\end{definition}

The previous definition allows us to define the action $\pi(g)$ on the vector of edge components of the function $f\in W^{2,2}(\Gamma)$ as
\begin{equation}
  \pi(g) \begin{pmatrix}f_{e_1}\\f_{e_2}\\ \vdots\\f_{e_{|\mathcal{E}|}}\end{pmatrix} = \begin{pmatrix}f_{g^{-1}*e_1}\\f_{g^{-1}*e_2}\\ \vdots\\f_{g^{-1}*e_{|\mathcal{E}|}}\end{pmatrix}\,.\label{eq:pi}
\end{equation}

The following definition of the kernel space will be useful for defining the quotient graph.

\begin{definition}
Let $\Gamma$ be a graph with the symmetry given by the symmetry group $(G,*)$, let $\pi$ be the permutation representation defined by \eqref{eq:pi} and let $\rho$ be an irreducible representation of $G$ with the dimension $r$. Then the kernel space associated with $\rho$ is defined as
\begin{equation}
  K_G(\rho,\pi) := \bigcap_{g\in G} \mathrm{Ker}\,[I_r \otimes \pi(g)-\rho(g)^\mathrm{T}\otimes I_{|\mathcal{E}|}]\,. \label{eq:kernel}
\end{equation}
Here, $\mathrm{Ker}$ denotes the kernel of the space in the parentheses, $I_r$ and $I_{|\mathcal{E}|}$ the $r\times r$ and $|\mathcal{E}| \times |\mathcal{E}|$ identity matrices, respectively, and $\mathrm{T}$ the transpose of a matrix.
\end{definition}

The following, slightly technical definition, introduces matrices needed in quotient graph construction.

\begin{definition}
We define the orbits $O_i := \{e_j\in \mathcal{P}: e_j = g*e_i \ \mathrm{for\ some}\ g\in G\}$. Let $\mathbf{e}_j$ be the standard basis of vectors in $\mathbb{C}^{|\mathcal{E}|}$ (do not confuse with the edges $e_j$). We define the space $X_i$ as the span of $\{\mathbf{e}_j: e_j\in O_i\}$. We define the set $\mathcal{D} = \{e_{j_1}, \dots, e_{j_{|\mathcal{D}|}}\}$ the set of edges so that each $e_{j_i}$ is one representative for each orbit $O_i$, hence $|\mathcal{D}|$ is the number of orbits. Let $V_\rho$ denote the carrier space of $\rho$. Then we define the subspaces $K_G^i(\rho,\pi) := K_G(\rho,\pi)\cap[V_\rho\otimes X_i]$, $i = 1, \dots, |\mathcal{D}|$. Let $d_i := \mathrm{dim\,}K_G^i(\rho,\pi)$ Let $\Theta_i$ be the matrices consisting of columns of vectors in the orthonormal basis of $K_G^i(\rho,\pi)$ for $i = 1, \dots, |\mathcal{D}|$. Finally, we define the matrices $\Theta := (\Theta_1, \Theta_2, \dots, \Theta_{|\mathcal{D}|})$ and $\hat \Theta := \Theta\otimes I_2$.
\end{definition}

Finally, we arrive at the definition of a quotient graph and at the main theorem.

\begin{definition}\label{def:band}
Let $\Gamma$ be a finite quantum graph with $|\mathcal{E}|$ edges and the coupling conditions \eqref{eq:qg:coup2} given by the matrices $A$ and $B$, which has the symmetry given by the group $(G,*)$. Then the quotient graph Hamiltonian $H_\rho$ corresponding to irreducible representation $\rho$ contained in the representation $\pi$ of $(G,*)$ of dimension $r$ is defined in the following way. It is given by the operator acting as negative second derivative on a graph $\Gamma_\rho$ consisting of the edges $\{e_{i,j}\}$ with $i\in \mathcal{D}$, $j = 1, \dots, d_i$ of the length $\ell_i$ (the edge length of the former edge $e_i$). The domain of the Hamiltonian on $\Gamma_\rho$ are functions in the Sobolev space $W^{2,2}(\Gamma_\rho)$ satisfying the coupling conditions given by the matrices
$$
  A_\rho := \hat \Theta^\dagger [I_r\otimes \tilde A] \hat\Theta\,,\quad B_\rho := \hat \Theta^\dagger [I_r\otimes \tilde B] \hat\Theta\,,
$$
where $\dagger$ denotes the hermitian conjugation and $\tilde A := (A+iB)^{-1}A$, $\tilde B := (A+iB)^{-1}B$. 
\end{definition}

\begin{theorem}\label{thm:band} (Band, Berkolaiko, Joyner, Liu)\\
The original Hamiltonian $H$ on the graph $\Gamma$ is unitarily equivalent to the direct sum over all irreducible representations of $G$ contained in $\pi$.
$$
  H \cong \bigoplus_{\rho}  H_\rho^{\oplus r(\rho)}.
$$
Here $r(\rho)$ is the dimension of the representation $\rho$ and $H_\rho^{\oplus r(\rho)}$ denotes $r(\rho)$ copies of the quotient graph operator $H_\rho$.
\end{theorem}

The second part of the theorem says that one has to take $r$ copies of the quotient graph corresponding to the representation $\rho$.

\section{The three-edge graph}\label{sec:three-edge}
We will apply the method introduced in the previous section to a particular graph. We consider a star graph consisting of three edges of the same length~$\ell$ with the same boundary conditions at the loose ends and a symmetric coupling condition at the central vertex. Later, we will introduce the coupling conditions; we will consider Neumann boundary conditions at the loose ends and three versions of the coupling condition at the central vertex. However, the first two quantum graphs are symmetric under the group $S_3$ -- the group of permutations of three elements; the third one has $C_3$ symmetry. 

\subsection{Representations of the group $S_3$}
Let us start by describing the group $S_3$. It consists of six elements: the identity element is the permutation that keeps all the edges, there are three permutations interchanging two edges and two which cyclically interchange all three edges. We will employ the notation $[ijk]$ for a permutation $g$ for which $g(1) = i$, $g(2) = j$, $g(3) = k$. In Table~\ref{tab:elements} we list all the permutations (group elements of the considered group $(G,*)$) and their inverse elements. There are three conjugacy classes, consisting of one, three, and two elements; in the table, these conjugacy classes are separated by a double vertical line.

\begin{table}[h]
\begin{tabular}{|c||c||c|c|c||c|c|}
\hline
$g$ & [123] & [213] & [321] & [132] & [231] & [312]\\ \hline
$g^{-1}$ & [123] & [213] & [321] & [132] & [312] & [231]\\ \hline
\end{tabular}
\caption{Elements of the group $S_3$ and their inverses.}
\label{tab:elements}
\end{table}

We give the form of the representation $\pi$ defined by \eqref{eq:pi}. The representation $\pi$ at each element is a $3\times 3$ permutation matrix. One can notice that for an element $[ijk]$, ones are in the first column and $i$-th row, second column and $j$-th row and in the third column and the $k$-th row, the other entries of the matrix are zero. The representation $\pi$ is for the group $S_3$ usually called the \emph{defining representation}.

\begin{eqnarray}
\pi([123]) = \begin{pmatrix}1 & 0 & 0\\ 0 & 1 & 0\\ 0 & 0 & 1\end{pmatrix}\,,\quad \pi([213]) = \begin{pmatrix}0 & 1 & 0\\ 1 & 0 & 0\\ 0 & 0 & 1\end{pmatrix}\,,\quad \pi([321]) = \begin{pmatrix}0 & 0 & 1\\ 0 & 1 & 0\\ 1 & 0 & 0\end{pmatrix}\,, \nonumber \\
\pi([132]) = \begin{pmatrix}1 & 0 & 0\\ 0 & 0 & 1\\ 0 & 1 & 0\end{pmatrix}\,,\quad \pi([231]) = \begin{pmatrix}0 & 0 & 1\\ 1 & 0 & 0\\ 0 & 1 & 0 \end{pmatrix}\,,\quad \pi([312]) = \begin{pmatrix}0 & 1 & 0\\ 0 & 0 & 1\\ 1 & 0 & 0\end{pmatrix}\,.\label{eq:defrep}
\end{eqnarray}

We leave for the reader to check that this really is a representation, i.e. that $\pi(g_1 * g_2) = \pi(g_1)\cdot \pi(g_2)$, where star denotes the group operation in the group $G$ (composition of permutations) and dot denotes matrix multiplication.

The next step will be finding the irreducible representations of the group $S_3$. Although the procedure can be found in the literature (e.g. \cite{FultonHarris}), for the reader's convenience we state it here as well.  Any group has the so-called \emph{trivial representation}, which is a one-dimensional representation assigning to all the elements number 1. One can easily prove that the one-dimensional representation assigning 1 to all even permutations and $-1$ to all odd permutations is also a representation of the group $S_3$. We will call it the \emph{signum representation}. The most difficult task will be to find the third irreducible representation, the \emph{orthogonal representation}, later we will find that it is a two-dimensional one.

First, we give the table of characters of the representations (see Table~\ref{tab:characters}). For the defining representation, the characters are obtained as the traces of the matrices in \eqref{eq:defrep}. For the one-dimensional representations (the trivial and signum representations) the characters are identical with the $1\times 1$ matrices of the representations. Notice that according to Proposition~\ref{thm:conjclass} the characters of the elements in the same conjugacy class are the same. In the following paragraphs, we comment on how the characters of the elements of the orthogonal representation are obtained and therefore how the last row of Table~\ref{tab:characters} is found.

\begin{table}[h]
\begin{tabular}{|c||c||c|c|c||c|c|}
\hline
group element & [123] & [213] & [321] & [132] & [231] & [312]\\ \hline
defining representation & 3 & 1 & 1 & 1 & 0 & 0\\ \hline
trivial representation & 1 & 1 & 1 & 1 & 1 & 1\\ \hline
signum representation & 1 & $-1$ & $-1$ & $-1$ & 1 & 1\\ \hline
orthogonal representation & 2 & 0 & 0 & 0 & $-1$ & $-1$\\ \hline
\end{tabular}
\caption{Characters of the representations.}
\label{tab:characters}
\end{table}

Let us now show that the trivial and signum representations are irreducible and that the defining representation is not. We will use Proposition~\ref{thm:inner}. The inner products are
\begin{eqnarray*}
\left<\chi_{\mathrm{def}},\chi_{\mathrm{def}}\right> &=& \frac{1}{6}(1\cdot 3\cdot 3+3\cdot 1\cdot 1+2\cdot 0 \cdot 0) = 2\,,\\
\left<\chi_{\mathrm{triv}},\chi_{\mathrm{triv}}\right> &=& \frac{1}{6}(1\cdot 1\cdot 1+3\cdot 1\cdot 1+2\cdot 1 \cdot 1) = 1\,,\\
\left<\chi_{\mathrm{sign}},\chi_{\mathrm{sign}}\right> &=& \frac{1}{6}(1\cdot 1\cdot 1+3\cdot (-1)\cdot (-1)+2\cdot 1 \cdot 1) = 1\,.
\end{eqnarray*}
The inner products for the trivial and signum representations are equal to 1, therefore these representations are irreducible, the defining representation is not. 

Now we find the multiplicities of the trivial and signum representations in the defining representation.
\begin{eqnarray*}
\left<\chi_{\mathrm{def}},\chi_{\mathrm{triv}}\right> &=& \frac{1}{6}(1\cdot 3\cdot 1+3\cdot 1\cdot 1+2\cdot 0 \cdot 1) = 1\,,\\
\left<\chi_{\mathrm{def}},\chi_{\mathrm{sign}}\right> &=& \frac{1}{6}(1\cdot 3\cdot 1+3\cdot 1\cdot (-1)+2\cdot 0 \cdot 1) = 0\,.
\end{eqnarray*}
We can see that the multiplicity of the trivial representation in the defining representation is 1, while the signum representation is not contained in the defining representation. Hence we define the orthogonal representation as the complement of the trivial representation in the defining representation and we have $\chi_{\mathrm{orth}} = \chi_{\mathrm{def}} - \chi_{\mathrm{triv}}$. This equation gives the last row in Table~\ref{tab:characters}. One can simply verify that the orthogonal representation is irreducible.
$$
\left<\chi_{\mathrm{orth}},\chi_{\mathrm{orth}}\right> = \frac{1}{6}(1\cdot 2\cdot 2+3\cdot 0\cdot 0+2\cdot (-1) \cdot (-1)) = 1\,.
$$

We proceed by finding the matrices of the orthogonal representation; the procedure was described, e.g., in \cite{unapologetic, Diaconis}. Since the trivial representation is contained in the defining representation with multiplicity one, we write the defining representation on a certain basis of the orthogonal complement of the subspace corresponding to the trivial representation. We use the following basis of the $\mathbb{R}^3$ space.
$$
  \mathbf{f}_1 = \mathbf{e}_1 + \mathbf{e}_2 + \mathbf{e}_3 = \begin{pmatrix}1\\1\\1\end{pmatrix}\,,\quad \mathbf{f}_2 = \mathbf{e}_2 - \mathbf{e}_1 = \begin{pmatrix}-1\\ 1\\ 0\end{pmatrix}\,,\quad  \mathbf{f}_2 = \mathbf{e}_3 - \mathbf{e}_1 = \begin{pmatrix}-1\\ 0\\ 1\end{pmatrix}\,.
$$

Let us show the construction for the group element $[321]$, which interchanges the first and the third edge.
\begin{eqnarray*}
\pi([321]) \mathbf{f}_1 &=& \pi([321])(\mathbf{e}_1+\mathbf{e}_2+\mathbf{e}_3) = \begin{pmatrix}0 & 0 & 1\\ 0 & 1 & 0\\ 1 & 0 & 0\end{pmatrix} \begin{pmatrix}1\\1\\1\end{pmatrix} = \begin{pmatrix}1\\1\\1\end{pmatrix} = \mathbf{e}_1+\mathbf{e}_2+\mathbf{e}_3 = \mathbf{f}_1\,.\\
\pi([321]) \mathbf{f}_2 &=& \pi([321])(\mathbf{e}_2-\mathbf{e}_1) = \begin{pmatrix}0 & 0 & 1\\ 0 & 1 & 0\\ 1 & 0 & 0\end{pmatrix} \begin{pmatrix}-1\\1\\0\end{pmatrix} = \begin{pmatrix}0\\1\\-1\end{pmatrix} =\\
&=& \mathbf{e}_2-\mathbf{e}_3 = \mathbf{e}_2- \mathbf{e}_1 - (\mathbf{e}_3- \mathbf{e}_1) = \mathbf{f}_2-\mathbf{f}_3\,.\\
\pi([321]) \mathbf{f}_3 &=& \pi([321])(\mathbf{e}_3-\mathbf{e}_1) = \begin{pmatrix}0 & 0 & 1\\ 0 & 1 & 0\\ 1 & 0 & 0\end{pmatrix} \begin{pmatrix}-1\\0\\1\end{pmatrix} = \begin{pmatrix}1\\0\\-1\end{pmatrix} = \mathbf{e}_1-\mathbf{e}_3 = -\mathbf{f}_3\,.\\
\end{eqnarray*}

If we write the action of this group element in the basis $\mathbf{f}_1$, $\mathbf{f}_2$, $\mathbf{f}_3$, we obtain the matrix $\begin{pmatrix}1 & 0 & 0\\ 0 & 1 & 0\\ 0 & -1 & -1\end{pmatrix}$. Since the subspace corresponding to the trivial representation is the span of $\mathbf{f}_1$, the restriction to the orthogonal space $\mathrm{span\,}\{\mathbf{f}_2, \mathbf{f}_3\}$ gives $\rho_{\mathrm{orth}}([321]) = \begin{pmatrix}1 & 0 \\ -1 & -1\end{pmatrix}$. Similarly, we can obtain the other matrices of the two-dimensional orthogonal representation.

\begin{eqnarray*}
\rho_{\mathrm{orth}}([123]) = \begin{pmatrix}1 & 0\\ 0 & 1 \end{pmatrix}\,,\quad \rho_{\mathrm{orth}}([213]) = \begin{pmatrix}-1 & -1\\ 0 & 1 \end{pmatrix}\,,\quad \rho_{\mathrm{orth}}([321]) = \begin{pmatrix}1 & 0 \\ -1 & -1 \end{pmatrix}\,,\\ \nonumber
\rho_{\mathrm{orth}}([132]) = \begin{pmatrix} 0 & 1\\ 1 & 0\end{pmatrix}\,,\quad \rho_{\mathrm{orth}}([231]) = \begin{pmatrix}-1 & -1\\ 1 & 0 \end{pmatrix}\,,\quad \rho_{\mathrm{orth}}([312]) = \begin{pmatrix}0 & 1 \\ -1 & -1\end{pmatrix}\,.
\end{eqnarray*}

\subsection{Application of the quotient graph theory}\label{sec:procedure}
In this subsection, we obtain the kernel space $K_{G}(\rho,\pi)$ and the corresponding matrices $\Theta$ and $\hat \Theta$ corresponding to irreducible representations of the group $S_3$. This part of the quotient graph theory does not depend on the coupling conditions, only on the symmetry group. 

Let us start with the orthogonal representation. We obtain the kernel space according to the equation~\eqref{eq:kernel}. Since the representation is two-dimensional, we will use $r=2$. The number of edges of the graph is $|\mathcal{E}| = 3$. We will show the construction of the kernel for the group element $[321]$ in detail and then list the results for other elements.
\begin{eqnarray*}
  I_2 \otimes \pi([321]) - \rho_{\mathrm{orth}}^\mathrm{T}([321])\otimes I_3 &=& \begin{pmatrix}1& 0\\ 0 &1\end{pmatrix} \cdot \begin{pmatrix}0 & 0 & 1\\ 0& 1 & 0\\ 1 & 0 & 0\end{pmatrix} - \begin{pmatrix}1 & -1\\ 0 & -1\end{pmatrix} \cdot \begin{pmatrix}1 & 0 & 0\\ 0 & 1 & 0\\ 0 & 0 & 1\end{pmatrix} = 
\\
&=& \begin{pmatrix}0 & 0 & 1 & 0 & 0 & 0\\0 & 1 & 0 & 0 & 0 & 0\\ 1 & 0 & 0 & 0 & 0 & 0\\ 0 & 0 & 0 & 0 & 0 & 1\\ 0 & 0 & 0 & 0 & 1 & 0\\ 0 & 0 & 0 & 1 & 0 & 0\end{pmatrix} - \begin{pmatrix}1 & 0 & 0 & -1 & 0 & 0\\0 & 1 & 0 & 0 & -1 & 0\\0 & 0 & 1 & 0 & 0 & -1\\ 0 & 0 & 0 & -1 & 0 & 0\\0 & 0 & 0 & 0 & -1 & 0\\0 & 0 & 0 & 0 & 0 & -1\end{pmatrix} =
\\
& = &  \begin{pmatrix}-1 & 0 & 1 & 1 & 0 & 0\\ 0 & 0 & 0 & 0 & 1 & 0\\ 1 & 0 & -1 & 0 & 0 & 1\\ 0 & 0 & 0 & 1 & 0 & 1\\ 0 & 0 & 0 & 0 & 2 & 0\\  0 & 0 & 0 & 1 & 0 & 1\end{pmatrix}\,.
\end{eqnarray*}
Hence we find that $\mathrm{Ker\,}(I_2 \otimes \pi([321]) - \rho_{\mathrm{orth}}^\mathrm{T}([321])\otimes I_3)$ is composed of the vectors $(a_1, a_2, \dots, a_6)^{\mathrm{T}}$ that satisfy 
$$
  a_5 = 0\,,\quad a_4 = -a_6\,,\quad -a_1 + a_3+ a_4 = 0\,.
$$

If we write down the conditions for the above kernels for the other group elements, we find that the space $K_{G}(\rho_{\mathrm{orth}},\pi)$ consists of vectors $(a_1, a_2, \dots, a_6)^{\mathrm{T}}$ satisfying 
$$
  a_1 = -a_2 = a_4 = -a_6\,,\quad a_3 = 0\,,\quad a_5=0\,.
$$
Therefore, the kernel space is the span of the vector $(1,-1,0, 1, 0 -1)^\mathrm{T}$. To obtain the matrix $\Theta$, we have to normalize this vector. We have
$$
  \Theta_{\mathrm{orth}} = \begin{pmatrix}\frac{1}{2}\\-\frac{1}{2}\\0\\\frac{1}{2}\\0\\-\frac{1}{2}\end{pmatrix}\,,\quad \hat \Theta_{\mathrm{orth}} = \Theta_{\mathrm{orth}}\otimes I_2 = \begin{pmatrix}\frac{1}{2}\\-\frac{1}{2}\\0\\\frac{1}{2}\\0\\-\frac{1}{2}\end{pmatrix} \otimes \begin{pmatrix}1 & 0 \\ 0 & 1\end{pmatrix} = \begin{pmatrix}\frac{1}{2} & 0\\ 0 & \frac{1}{2}\\- \frac{1}{2} & 0 \\ 0 & -\frac{1}{2}\\0 & 0 \\ 0 & 0 \\ \frac{1}{2} & 0 \\ 0 & \frac{1}{2}\\ 0 & 0 \\ 0 & 0\\ -\frac{1}{2} & 0 \\ 0 &- \frac{1}{2}\end{pmatrix}\,.
$$

Now we continue with the trivial representation. Now $r=1$. Let us again show the construction for the group element $[321]$. 
\begin{eqnarray*}
  I_1 \otimes \pi([321]) - \rho_{\mathrm{triv}}^\mathrm{T}([321])\otimes I_3 &=& 1\otimes \begin{pmatrix}0 & 0 & 1\\ 0 & 1 & 0\\ 1 & 0 & 0\end{pmatrix} - 1\otimes \begin{pmatrix}1 & 0 & 0\\ 0 & 1 & 0\\ 0 & 0 & 1\end{pmatrix} = \begin{pmatrix}-1 & 0 & 1\\ 0 & 0 & 0 \\ 1 & 0 & -1\end{pmatrix}\,.
\end{eqnarray*}
Hence $\mathrm{Ker\,}(I_1 \otimes \pi([321]) - \rho_{\mathrm{triv}}^\mathrm{T}([321])\otimes I_3)$ consists of vectors $(a_1, a_2, a_3)^{\mathrm{T}}$ satisfying $a_1 -a_3 = 0$.

From the other group elements we obtain equations $a_1 = a_2$ and $a_2 = a_3$, thus resulting into equation $a_1 = a_2 = a_3$ which describes the vectors in $K_G(\rho_{\mathrm{triv}},\pi)$. This kernel space thus is the span of the vector $(1,1,1)^\mathrm{T}$. After normalization we obtain
$$
  \Theta_{\mathrm{triv}} = \begin{pmatrix}\frac{1}{\sqrt{3}}\\ \frac{1}{\sqrt{3}}\\ \frac{1}{\sqrt{3}}\end{pmatrix}\,,\quad \hat \Theta_{\mathrm{triv}} = \Theta_{\mathrm{triv}}\otimes I_2 = \begin{pmatrix}\frac{1}{\sqrt{3}}\\ \frac{1}{\sqrt{3}}\\ \frac{1}{\sqrt{3}}\end{pmatrix}\otimes \begin{pmatrix}1& 0 \\ 0 & 1\end{pmatrix} =  \begin{pmatrix}\frac{1}{\sqrt{3}} & 0 \\ 0 & \frac{1}{\sqrt{3}}\\ \frac{1}{\sqrt{3}} & 0 \\ 0 & \frac{1}{\sqrt{3}}\\ \frac{1}{\sqrt{3}} & 0 \\ 0 & \frac{1}{\sqrt{3}}\end{pmatrix}\,.
$$

By a similar procedure, it can be proven that the kernel space for the signum representation is empty. This is connected to the fact that the signum representation is not contained in the defining representation. 

\subsection{Standard condition at the central vertex}\label{sec:standard}
\begin{figure}
\centering
\begin{subfigure}{.3\textwidth}
  \centering\captionsetup{width=.9\linewidth}
  \includegraphics[width=.99\linewidth]{fig1a}
  \caption{Standard coupling}
  \label{fig1a}
\end{subfigure}%
\hspace{0.03\textwidth}
\begin{subfigure}{.3\textwidth}
  \centering\captionsetup{width=.9\linewidth}
  \includegraphics[width=.99\linewidth]{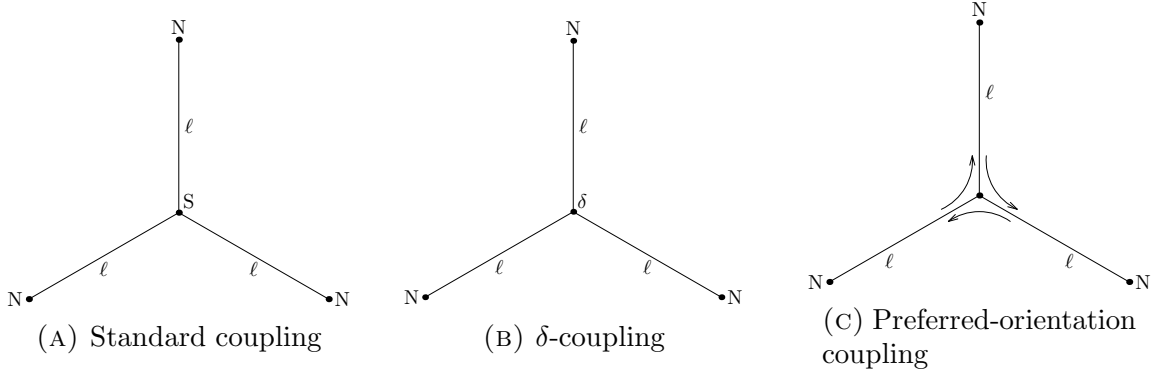}
  \caption{$\delta$-coupling}
  \label{fig1b}
\end{subfigure}
\hspace{0.03\textwidth}
\begin{subfigure}{.3\textwidth}
  \centering\captionsetup{width=.9\linewidth}
  \includegraphics[width=.99\linewidth]{fig1c}
  \caption{Preferred-orientation coupling}
  \label{fig1c}
\end{subfigure}%
\caption{Figures of the three-edge graphs considered in Subsections~\ref{sec:standard}, \ref{sec:delta}, and \ref{sec:preferred}.}
\label{fig1}
\end{figure}

Let us first consider a quantum star graph consisting of three edges of the length~$\ell$ (see Fig.~\ref{fig1a}). We parametrize the edges by the intervals $(0,\ell)$ with $x= 0$ at the loose ends and $x = \ell$ at the central vertex. We assume Neumann boundary conditions at the loose ends and standard coupling at the central vertex. 
\begin{equation}
  f_1'(0) = f_2'(0) = f_3'(0) = 0\,,\quad f_1(\ell) = f_2(\ell) = f_3(\ell) \,,\quad -f_1'(\ell)-f_2'(\ell)-f_3'(\ell)=0\,. \label{eq:ccstandard}
\end{equation}
The matrices $A$ and $B$ corresponding to these coupling conditions are
$$
  A = \begin{pmatrix}0 & 0 & 0 & 0 & 0 & 0\\0 & 0 & 0 & 0 & 0 & 0\\0 & 0 & 0 & 0 & 0 & 0\\0 & 1 & 0 & -1 & 0 & 0\\0 & 1 & 0 & 0 & 0 & -1\\0 & 0 & 0 & 0 & 0 & 0\end{pmatrix}\,,\quad B = \begin{pmatrix}1 & 0 & 0 & 0 & 0 & 0\\0 & 0 & 1 & 0 & 0 & 0\\0 & 0 & 0 & 0 & 1 & 0\\0 & 0 & 0 & 0 & 0 & 0\\0 & 0 & 0 & 0 & 0 & 0\\0 & 1 & 0 & 1 & 0 & 1\end{pmatrix}\,.
$$
Hence we obtain
$$
  \tilde A = \begin{pmatrix}0 & 0 & 0 & 0 & 0 & 0\\0 & \frac{2}{3} & 0 & -\frac{1}{3} & 0 & -\frac{1}{3}\\0 & 0 & 0 & 0 & 0 & 0\\0 & -\frac{1}{3} & 0 & \frac{2}{3} & 0 & -\frac{1}{3}\\0 & 0 & 0 & 0 & 0 & 0\\0 & -\frac{1}{3} & 0 & -\frac{1}{3} & 0 & \frac{2}{3}\end{pmatrix}\,,\quad \tilde B = \begin{pmatrix}-i & 0 & 0 & 0 & 0 & 0\\0 & -\frac{i}{3} & 0 &  -\frac{i}{3} & 0 &  -\frac{i}{3}\\0 & 0 & -i & 0 & 0 & 0\\0 & -\frac{i}{3} & 0 &  -\frac{i}{3} & 0 &  -\frac{i}{3}\\0 & 0 & 0 & 0 & -i & 0\\0 & -\frac{i}{3} & 0 &  -\frac{i}{3} & 0 &  -\frac{i}{3}\end{pmatrix} \,.
$$

Using Definition~\ref{def:band}  we obtain for the orthogonal representation
$$
  A_{\rho_\mathrm{orth}} = \hat \Theta_{\mathrm{orth}}^\dagger (I_2\otimes \tilde A) \Theta_{\mathrm{orth}} =\begin{pmatrix}0 & 0 \\ 0 & 1\end{pmatrix}\,,\quad  B_{\rho_\mathrm{orth}} = \hat \Theta_{\mathrm{orth}}^\dagger (I_2\otimes \tilde B) \Theta_{\mathrm{orth}} =\begin{pmatrix}-i & 0 \\ 0 & 0\end{pmatrix}\,.
$$
and for the trivial representation
$$
  A_{\rho_\mathrm{triv}} = \hat \Theta_{\mathrm{triv}}^\dagger  \tilde A \Theta_{\mathrm{triv}} =\begin{pmatrix}0 & 0 \\ 0 & 0\end{pmatrix}\,,\quad  B_{\rho_\mathrm{triv}} = \hat \Theta_{\mathrm{triv}}^\dagger \tilde B \Theta_{\mathrm{triv}} =\begin{pmatrix}-i & 0 \\ 0 & -i\end{pmatrix}\,.
$$

The graph $\Gamma_\rho$ is for the trivial representation the segment $(0,\ell)$; for the orthogonal representation we obtain two copies of this segment. The coupling conditions of the graphs $\Gamma_\rho$ are given by condition~\eqref{eq:qg:coup2} with $\Psi = \begin{pmatrix}f(0)\\f(\ell)\end{pmatrix}$ and $\Psi' = \begin{pmatrix}f'(0)\\-f'(\ell)\end{pmatrix}$, where $f$ denotes the wavefunction on the segment. For the orthogonal representation, the coupling condition~\eqref{eq:qg:coup2} with the coupling matrices $A_{\rho_\mathrm{orth}}$ and $B_{\rho_\mathrm{orth}}$ gives $f'(0) = 0$ and $f(\ell) = 0$, i.e. the Neumann boundary condition at one end and Dirichlet at the other. There are two copies of this graph since the dimension of the representation is two. The coupling matrices for the trivial representation follow from the coupling condition~\eqref{eq:qg:coup2} with the matrices $A_{\rho_\mathrm{triv}}$ and $B_{\rho_\mathrm{triv}}$. We obtain $f'(0)=0$ and $f'(\ell) =0$, which corresponds to Neumann boundary conditions at both ends of the segment. Since the kernel space for the signum representation is trivial, the graph $\Gamma_\rho$ is in this case empty. 

\subsection{$\delta$-condition at the central vertex}\label{sec:delta}
In the second example, we consider the same graph as in Subsection~\ref{sec:standard}, only the coupling condition at the central vertex is replaced by the so-called $\delta$-condition of the strength $\alpha\in\mathbb{R}$ (see Fig.~\ref{fig1b}). 
$$
  f_1'(0) = f_2'(0) = f_3'(0) = 0\,,\quad f_1(\ell) = f_2(\ell) = f_3(\ell) \,,\quad -f_1'(\ell)-f_2'(\ell)-f_3'(\ell)=\alpha f_1(\ell)\,. 
$$

The corresponding coupling matrices read as follows.
$$
  A = \begin{pmatrix}0 & 0 & 0 & 0 & 0 & 0\\0 & 0 & 0 & 0 & 0 & 0\\0 & 0 & 0 & 0 & 0 & 0\\0 & 1 & 0 & -1 & 0 & 0\\0 & 1 & 0 & 0 & 0 & -1\\0 & -\alpha & 0 & 0 & 0 & 0\end{pmatrix}\,,\quad B = \begin{pmatrix}1 & 0 & 0 & 0 & 0 & 0\\0 & 0 & 1 & 0 & 0 & 0\\0 & 0 & 0 & 0 & 1 & 0\\0 & 0 & 0 & 0 & 0 & 0\\0 & 0 & 0 & 0 & 0 & 0\\0 & 1 & 0 & 1 & 0 & 1\end{pmatrix}\,.
$$
Hence we obtain
$$
  \tilde A = \begin{pmatrix}0 & 0 & 0 & 0 & 0 & 0 \\ 0 & \frac{\alpha-2i}{\alpha-3i} & 0 & \frac{i}{\alpha-3i} & 0 & \frac{i}{\alpha-3i} \\0 & 0 & 0 & 0 & 0 & 0 \\ 0 & \frac{i}{\alpha-3i} & 0 & \frac{\alpha-2i}{\alpha-3i} & 0 & \frac{i}{\alpha-3i}\\0 & 0 & 0 & 0 & 0 & 0 \\ 0 & \frac{i}{\alpha-3i} & 0 & \frac{i}{\alpha-3i} & 0 & \frac{\alpha-2i}{\alpha-3i} \end{pmatrix}\,,\quad \tilde B = \begin{pmatrix}-i & 0 & 0 & 0 & 0 & 0 \\ 0 & \frac{-1}{\alpha-3i} & 0 & \frac{-1}{\alpha-3i} & 0 & \frac{-1}{\alpha-3i} \\0 & 0 & -i & 0 & 0 & 0 \\ 0 & \frac{-1}{\alpha-3i} & 0 & \frac{-1}{\alpha-3i} & 0 & \frac{-1}{\alpha-3i}\\0 & 0 & 0 & 0 & -i & 0 \\ 0 & \frac{-1}{\alpha-3i} & 0 & \frac{-1}{\alpha-3i} & 0 & \frac{-1}{\alpha-3i}\end{pmatrix} \,.
$$

From Definition~\ref{def:band} we obtain for the orthogonal representation
$$
  A_{\rho_\mathrm{orth}} = \hat \Theta_{\mathrm{orth}}^\dagger (I_2\otimes \tilde A) \Theta_{\mathrm{orth}} =\begin{pmatrix}0 & 0 \\ 0 & 1\end{pmatrix}\,,\quad  B_{\rho_\mathrm{orth}} = \hat \Theta_{\mathrm{orth}}^\dagger (I_2\otimes \tilde B) \Theta_{\mathrm{orth}} =\begin{pmatrix}-i & 0 \\ 0 & 0\end{pmatrix}\,.
$$
and for the trivial representation
$$
  A_{\rho_\mathrm{triv}} = \hat \Theta_{\mathrm{triv}}^\dagger  \tilde A \Theta_{\mathrm{triv}} =\begin{pmatrix}0 & 0 \\ 0 & \frac{\alpha}{\alpha-3i}\end{pmatrix}\,,\quad  B_{\rho_\mathrm{triv}} = \hat \Theta_{\mathrm{triv}}^\dagger \tilde B \Theta_{\mathrm{triv}} =\begin{pmatrix}-i & 0 \\ 0 &\frac{-3}{\alpha-3i}\end{pmatrix}\,.
$$

Similarly to the previous example, the orthogonal representation gives two copies of the quotient graph with Neumann boundary condition at one end and Dirichlet at the other. The trivial representation leads to the segment with Neumann boundary condition at one end and Robin boundary condition with the coupling parameter $\alpha/3$
$$
  \alpha g(\ell)-3 (-g'(\ell)) = 0\quad \Rightarrow \quad -g'(\ell) = \frac{\alpha}{3} g(\ell)
$$
at the other end. The signum representation gives, as in the previous section, the empty graph.

\subsection{Preferred-orientation coupling at the central vertex}\label{sec:preferred}
In the last example, we consider the coupling condition of preferred orientation at the central vertex, earlier studied in \cite{ExnerTater, ExnerLipovsky1, ExnerLipovsky2, Baradaran}. This coupling condition, motivated by application to modeling quantum Hall effect was first used in \cite{ExnerTater}. For the particular energy $E = 1$ the wave coming from one edge is fully transmitted to the neighbouring edge, the wave coming from this edge is fully transmitted to the next edge, etc. cyclically (see Figure~\ref{fig1c}). It was found that the transport properties of the preferred orientation coupling depend on the parity of the vertex (i.e. whether the vertex degree is even or odd). The vertex coupling matrices are $A_v = U_v-I$ and $B_v = i(U_v+I)$ with $U = \begin{pmatrix}0 & 1 & 0\\ 0 & 0 & 1\\ 1 & 0 & 0\end{pmatrix}$. The boundary conditions at the loose ends will again be Neumann.

The coupling matrices of the whole graph are
$$
  A = \begin{pmatrix}0 & 0 & 0 & 0 & 0 & 0\\0 & 0 & 0 & 0 & 0 & 0\\0 & 0 & 0 & 0 & 0 & 0\\0 & -1 & 0 & 1 & 0 & 0\\ 0 & 0 & 0 & -1 & 0 & 1\\ 0 & 1 & 0 & 0 & 0 & -1\end{pmatrix}\,,\quad B = \begin{pmatrix}1 & 0 & 0 & 0 & 0 & 0\\0 & 0 & 1 & 0 & 0 & 0\\0 & 0 & 0 & 0 & 1 & 0\\0 & i & 0 & i & 0 & 0\\0 & 0 & 0 & i & 0 & i\\0 & i & 0 & 0 & 0 & i\end{pmatrix}\,.
$$

We have 
$$
  \tilde A = \begin{pmatrix}0 & 0 & 0 & 0 & 0 & 0 \\ 0 & \frac{1}{2} & 0 & -\frac{1}{2} & 0 & 0 \\ 0 & 0 & 0 & 0 & 0 & 0 \\0 & 0 & 0 & \frac{1}{2} & 0 & -\frac{1}{2} \\0 & 0 & 0 & 0 & 0 & 0 \\0 & -\frac{1}{2} & 0 & 0 & 0 & \frac{1}{2}  \end{pmatrix}\,,\quad \tilde B = \begin{pmatrix}-i & 0 & 0 & 0 & 0 & 0 \\ 0 & -\frac{i}{2} & 0 & -\frac{i}{2} & 0 & 0 \\0 & 0 & -i & 0 & 0 & 0 \\0 & 0 & 0 &  -\frac{i}{2} & 0 &  -\frac{i}{2} \\0 & 0 & 0 & 0 & -i & 0 \\0 & -\frac{i}{2} & 0 & 0 & 0 & -\frac{i}{2} \end{pmatrix} \,.
$$

However, one cannot use the same symmetry group as in the previous two examples. The graph is no longer symmetric with respect to the symmetry group $S_3$ since e.g. interchanging two edges would change the direction of the wave for $E=1$. From the former group $S_3$ only the elements $[123]$ (identity), $[231]$ and $[312]$ (cyclic permutations) do not change the symmetry of the graph (note that all these permutations are even). The symmetry of the graph is, therefore, $C_3$. It has three elements, the identity, the rotation (denoted by $a$) by the angle $2\pi/3$ and its inverse element $a^{-1}$, i.e. the rotation by the angle $-2\pi/3$. The group has three one-dimensional irreducible representations, its character table is given in Table~\ref{tab:c3}.

\begin{table}[h]
\begin{tabular}{|c||c||c||c|}
\hline
& 1 & $a$ & $a^{-1}$\\ \hline
$\chi_{1}$ & 1 & 1 & 1\\ \hline
$\chi_{2}$ & 1 & $\omega$ & $\bar\omega$\\ \hline
$\chi_{3}$ & 1 & $\bar\omega$ & $\omega$\\ \hline
\end{tabular}
\caption{Character table of the group $C_3$. Here, $\omega = \mathrm{e}^{2\pi i/3}$, $\bar\omega = \mathrm{e}^{-2\pi i/3}$.}
\label{tab:c3}
\end{table}

We proceed similarly as with the group $S_3$ -- we find the representation $\pi$ and the three irreducible representations that are identical to the above characters. Then we apply the procedure from Subsection~\ref{sec:procedure} to find the matrices $\Theta$ and $\hat \Theta$. Finally, we obtain coupling matrices of the quotient graphs $A_\rho$ and $B_\rho$. We list the results.

The representation $\pi$ is
$$
\pi(1) = \begin{pmatrix}1 & 0 & 0\\ 0 & 1 & 0\\ 0 & 0 & 1\end{pmatrix}\,,\quad \pi(a) = \begin{pmatrix}0 & 0 & 1\\ 1 & 0 & 0\\ 0 & 1 & 0\end{pmatrix}\,,\quad \pi(a^{-1}) = \begin{pmatrix}0 & 1 & 0\\ 0 & 0 & 1\\ 1 & 0 & 0\end{pmatrix}\,.\\
$$
The irreducible representations are
\begin{eqnarray*}
\rho_1(1) = 1\,,\quad & \rho_1(a) = 1\,,\quad & \rho_1(a^{-1})=1\,,\\
\rho_2(1) = 1\,,\quad & \rho_2(a) = \mathrm{e}^{2\pi i /3}\,,\quad & \rho_2(a^{-1})=\mathrm{e}^{-2\pi i /3}\,,\\
\rho_3(1) = 1\,,\quad & \rho_3(a) = \mathrm{e}^{-2\pi i /3}\,,\quad & \rho_3(a^{-1})=\mathrm{e}^{2\pi i /3}\,.\\  
\end{eqnarray*}

All the graphs $\Gamma_\rho$ are segments of the length~$\ell$. Below, we obtain their boundary conditions. For the first irreducible representation, we get
$$
  \Theta_1 = \begin{pmatrix}-\frac{1}{\sqrt{3}}\\-\frac{1}{\sqrt{3}}\\-\frac{1}{\sqrt{3}}\end{pmatrix}\,,\quad \hat \Theta_1 = \begin{pmatrix}-\frac{1}{\sqrt{3}} & 0 \\ 0 & -\frac{1}{\sqrt{3}}\\ -\frac{1}{\sqrt{3}} & 0\\ 0 & -\frac{1}{\sqrt{3}}\\ -\frac{1}{\sqrt{3}}& 0 \\ 0 & -\frac{1}{\sqrt{3}}\end{pmatrix}\,,\quad A_{\rho_1} = \begin{pmatrix}0 & 0 \\ 0 & 0\end{pmatrix}\,,\quad B_{\rho_1} = \begin{pmatrix}-i & 0\\ 0  & -i\end{pmatrix}\,.
$$
This corresponds to the Neumann boundary condition at both ends of the interval.

The second and third representations yield ($\omega = \mathrm{e}^{2\pi i/3}$, $\bar\omega = \mathrm{e}^{-2\pi i/3}$)
$$
  \Theta_2 = \begin{pmatrix}-\frac{1}{\sqrt{3}}\\-\frac{\bar \omega}{\sqrt{3}}\\-\frac{\omega}{\sqrt{3}}\end{pmatrix}\,,\quad \hat \Theta_2 = \begin{pmatrix}-\frac{1}{\sqrt{3}} & 0 \\ 0 & -\frac{1}{\sqrt{3}}\\ -\frac{\bar\omega}{\sqrt{3}} & 0\\ 0 & -\frac{\bar\omega}{\sqrt{3}}\\ -\frac{\omega}{\sqrt{3}}& 0 \\ 0 & -\frac{\omega}{\sqrt{3}}\end{pmatrix}\,,\quad A_{\rho_1} = \begin{pmatrix}0 & 0 \\ 0 & \frac{1}{2}(1-\bar\omega)\end{pmatrix}\,,\quad B_{\rho_1} = \begin{pmatrix}-i & 0\\ 0  & \frac{-i}{2}(1+\bar\omega)\end{pmatrix}\,.
$$
$$
  \Theta_3 = \begin{pmatrix}-\frac{1}{\sqrt{3}}\\-\frac{\omega}{\sqrt{3}}\\-\frac{\bar\omega}{\sqrt{3}}\end{pmatrix}\,,\quad \hat \Theta_3 = \begin{pmatrix}-\frac{1}{\sqrt{3}} & 0 \\ 0 & -\frac{1}{\sqrt{3}}\\ -\frac{\omega}{\sqrt{3}} & 0\\ 0 & -\frac{\omega}{\sqrt{3}}\\ -\frac{\bar\omega}{\sqrt{3}}& 0 \\ 0 & -\frac{\bar\omega}{\sqrt{3}}\end{pmatrix}\,,\quad A_{\rho_3} = \begin{pmatrix}0 & 0 \\ 0 & \frac{1}{2}(1-\omega)\end{pmatrix}\,,\quad B_{\rho_3} = \begin{pmatrix}-i & 0\\ 0  & \frac{-i}{2}(1+\omega)\end{pmatrix}\,.
$$
These coupling matrices correspond to the Neumann boundary condition at one end of the segment and Robin condition with the coefficient $\pm \sqrt{3}$ at the other. For the second representation, we have
$$
  g'(0) =0\,,\quad -g'(\ell) = \frac{1}{i}\frac{1-\bar\omega}{1+\bar\omega} g(\ell) = \sqrt{3}g(\ell) 
$$
and for the third
$$
  g'(0) =0\,,\quad -g'(\ell) = \frac{1}{i}\frac{1-\omega}{1+\omega} g(\ell) = -\sqrt{3}g(\ell) \,.
$$

\section{Conclusions}\label{sec:conclusions}
We have illustrated the usage of the quotient graph method on three-edge star graphs. For the graph with Neumann boundary condition at the loose ends and standard coupling at the central vertex, we obtained three segments of lengths~$\ell$, one with Neumann boundary condition at both ends, two with Neumann boundary condition at one end, and Dirichlet at the other end. For the graph with Neumann boundary condition at the loose ends and $\delta$-condition at the central vertex, we again obtained two copies of the segment of length~$\ell$ with Neumann and Dirichlet conditions at the opposite ends; the third quotient graph is a segment of length~$\ell$ with Neumann boundary condition at one end, and Robin condition (with the coupling parameter $\alpha/3$) at the other end. The example with the preferred-orientation coupling is symmetric under the $C_3$ symmetry group and its quotient graphs are the segments of length~$\ell$, one with Neumann condition at both ends, the two other with Neumann condition at one end and Robin (with the parameter $\pm \sqrt{3}$) at the other end.

We should stress that the above results can be obtained also without the machinery of~\cite{BandBerkolaikoJoynerLiu}. The trivial representation corresponds to the symmetric subspace of the domain of the Hamiltonian and the orthogonal representation (or, in the case of preferred-orientation coupling the representations $\chi_2$ and $\chi_3$) correspond to the two-dimensional subspace of antisymmetric functions. However, the current note can serve as a simple but non-trivial example of the quotient graph theory for quantum graphs and together with the original paper~\cite{BandBerkolaikoJoynerLiu} can teach the reader the procedures necessary for dealing with more complicated problems.

Finally, let us illustrate how the relation $H \cong \bigoplus_{\rho}  H_\rho^{\oplus r(\rho)}$ can be obtained in case of the graph with standard coupling at the central vertex. Let the wavefunction components of the three-edge graph be $f_1$, $f_2$, $f_3$. The domain of the Hamiltonian on the three-edge graph can be decomposed into the symmetric subspace (represented by $h_{\mathrm{sym}}(x) = \frac{1}{\sqrt{3}}(f_1(x)+f_2(x)+f_3(x))$ with $x\in (0,\ell)$ and corresponding to the trivial representation) and the two-dimensional anti-symmetric subspace (represented by $h_{\mathrm{ant1}}(x) = \frac{1}{\sqrt{2}}(f_1(x)-f_2(x))$ and $h_{\mathrm{ant2}}(x) = \frac{1}{\sqrt{2}}(f_1(x)-f_3(x))$ with $x\in (0,\ell)$, corresponding to the orthogonal representation). The coupling conditions on the three-edge graph~\eqref{eq:ccstandard} yield
\begin{eqnarray*}
  h_{\mathrm{sym}}'(0) &=& \frac{1}{\sqrt{3}}(f_1'(0)+f_2'(0)+f_3'(0)) = 0\,,\\
  h_{\mathrm{sym}}'(\ell) &=& \frac{1}{\sqrt{3}}(f_1'(\ell)+f_2'(\ell)+f_3'(\ell)) = 0\,,\\
  h_{\mathrm{ant1}}'(0) &=& \frac{1}{\sqrt{2}}(f_1'(0)-f_2'(0)) = 0\,,\\
  h_{\mathrm{ant1}}(\ell) &=& \frac{1}{\sqrt{2}}(f_1(\ell)-f_2(\ell)) = 0\,,\\
  h_{\mathrm{ant2}}'(0) &=& \frac{1}{\sqrt{2}}(f_1'(0)-f_3'(0)) = 0\,,\\
  h_{\mathrm{ant2}}(\ell) &=& \frac{1}{\sqrt{2}}(f_1(\ell)-f_3(\ell)) = 0\,.
\end{eqnarray*}
Therefore, we show that the symmetric subspace corresponds to the segment with Neumann boundary conditions at both ends and the antisymmetric subspace to two copies of the segment with Neumann condition at one end and Dirichlet at the other. The Hamiltonian on the former three-edge graph is unitarily equivalent to the orthogonal sum of the three mentioned operators.

\section*{Acknowledgements}
J.L. was supported by the Research Programme ``Mathematical Physics and Differential Geometry'' of the Faculty of Science of the University of Hradec Kr\'alov\'e. The authors thank Ram Band for the suggestions that improved the paper.


\begin{thebibliography}{99}
\bibitem{Pauling} L. Pauling, The diamagnetic anisotropy of aromatic molecules, \emph{J. Chem. Phys.} {\bf 4} (1936), pp. 673--677. DOI: 10.1063/1.1749766.

\bibitem{RuedenbergScherr} K. Ruedenberg, C. Scherr, Free-electron network model for conjugated systems, {I}. {T}heory, \emph{J. Chem. Phys.} {\bf 21} (1953), pp. 1565--1581. DOI: 10.1063/1.1699299.

\bibitem{KottosSmilansky} T. Kottos, U. Smilansky, Quantum chaos on graphs, \emph{Phys. Rev. Lett} {\bf 79} (1997), 4794--4797. DOI: 10.1103/PhysRevLett.79.4794. 

\bibitem{BerkolaikoBogomolny} G. Berkolaiko, E.B. Bogomolny, J.P. Keating, Star graphs and \v{S}eba billiards, \emph{J. Phys. A} {\bf 34} (2001), 335--350. DOI: 10.1088/0305-4470/34/3/301.

\bibitem{HulEtAl} O. Hul, S. Bauch, P. Pako\'nski, N. Savytskyy, K. \. Zyczkowski, and L. Sirko, Experimental simulation of quantum graphs by microwave networks, \emph{Phys. Rev. E} {\bf 69} (2004), 056205. DOI: 10.1103/PhysRevE.69.056205.

\bibitem{HulEtAl2} O. Hul, M. {\L}awniczak, S. Bauch, A. Sawicki, M. Ku\'s and L. Sirko, Are Scattering Properties of Graphs Uniquely Connected to Their Shapes?, \emph{Phys. Rev. Lett.} {\bf 109} (2012), 040402. DOI: 10.1103/PhysRevLett.109.040402.

\bibitem{LawniczakLipovskySirko} M. {\L}awniczak, J. Lipovsk\'y, and L. Sirko, Non-Weyl microwave graphs, \emph{Phys. Rev. Lett.} {\bf 122} (2019), 140503. DOI: 10.1103/PhysRevLett.122.140503.

\bibitem{LawniczakKurasov} M. {\L}awniczak, P. Kurasov, S. Bauch, M. Bia{\l}ous, V. Yunko, and L. Sirko, Hearing Euler characteristic of graphs, \emph{Phys. Rev. E} {\bf 101} (2020), 052320. DOI: 10.1103/PhysRevE.101.052320.

\bibitem{LawniczakLipovskyBialousSirko} M. {\L}awniczak, J. Lipovsk\'y, M. Bia{\l}ous, L. Sirko, Application of topological resonances in experimental investigation of a Fermi golden rule in microwave networks, \emph{Phys. Rev. E} {\bf 103} (2021), 032208. DOI: 10.1103/PhysRevE.103.032208.

\bibitem{BandBerkolaikoJoynerLiu} R. Band, G. Berkolaiko, C. H. Joyner, W. Liu, Quotients of finite-dimensional operators by symmetry representations, arXiv preprint, arXiv:1711.00918 [math-ph]. 

\bibitem{BandParzanchevskiBen-Shach} R. Band, O. Parzanchevski, G. Ben-Shach, The Isospectral Fruits of Representation Theory: Quantum Graphs and Drums, \emph{J. Phys. A: Math. Theor.} {\bf 42} (2009), 17520. DOI: 10.1088/1751-8113/42/17/175202.

\bibitem{BandParzanchevski} O. Parzanchevski, R. Band, Linear Representations and Isospectrality with Boundary Conditions, \emph{J. Geom. Anal.} {\bf 20} (2010), p. 439--471. DOI: 10.1007/s12220-009-9115-6.

\bibitem{ExnerTater} P. Exner, M. Tater, Quantum graphs with vertices of a preferred orientation, \emph{Phys. Lett. A} {\bf 382} (2018), pp. 283--287. DOI: j.physleta.2017.11.028.

\bibitem{ExnerLipovsky1} P. Exner, J. Lipovsk\'y, Spectral asymptotics of the Laplacian on Platonic solids graphs, \emph{J. Math. Phys.} {\bf 60} (2019), 122101. DOI: 10.1063/1.5116100.

\bibitem{JoynerMullerSieber} C. H. Joyner, S. M\"uller, and M. Sieber, GSE statistics without spin, \emph{EPL} \textbf{107} (2014), 50004. DOI: 10.1209/0295-5075/107/50004.

\bibitem{FultonHarris} W. Fulton, J. Harris, \emph{Representation Theory}, Graduate Texts in Mathematics 129, Springer, New York, 2004, 551 pp. ISBN: 978-1-4612-0979-9. DOI: 10.1007/978-1-4612-0979-9.

\bibitem{BarutRaczka} A. O. Barut, R. R\c{a}czka, \emph{Theory of Group Representations and Applications}, Default Book Series, World Scientific, 1986,
740 pp. ISBN: 9789971502164. DOI: 10.1142/0352 

\bibitem{Cotton} F. A. Cotton, \emph{Chemical Applications of Group Theory}, 3rd Edition, Wiley, New York, 1990, 461 pp. ISBN: 978-0-471-51094-9. 
  
\bibitem{BerkolaikoKuchment} G. Berkolaiko, P. Kuchment, \emph{Introduction to Quantum Graphs}, Mathematical Surveys and Monographs 186. AMS, 2013, 270 pp, ISBN 978-0-8218-9211-4. DOI: 10.1090/surv/186.

\bibitem{GnutzmannSmilansky} S. Gnutzmann, U. Smilansky,  Quantum graphs: Applications to quantum chaos and universal spectral statistics, \emph{Advances in Physics} {\bf 55} (2006), 527--625. DOI: 10.1080/00018730600908042 

\bibitem{unapologetic} The Unapologetic Mathematician, \emph{One Complete Character Table (part 2)}, blog post, \texttt{https://unapologetic.wordpress.com/2010/10/27/one-complete-character-table-part-2/}, published 2010, retrieved 27th June 2021.

\bibitem{Diaconis} P. Diaconis, \emph{Group representations in probability and statistics}, IMS Lecture Notes Monogr. Ser., 11, 1988, 198pp ISBN: 0940600145. DOI: 10.1214/lnms/1215467407. 

\bibitem{ExnerLipovsky2} P. Exner, J. Lipovsk\'y, Topological bulk-edge effects in quantum graph transport, \emph{Phys. Lett. A} {\bf 384} (2020), 126390. DOI: /10.1016/j.physleta.2020.126390.

\bibitem{Baradaran} M. Baradaran, P. Exner, M. Tater, Ring chains with vertex coupling of a preferred orientation, \emph{Rev. Math. Phys.} {\bf 32} (2020), 2060005. DOI: 10.1142/S0129055X20600053.
\end{thebibliography}
\end{document}